\documentclass[conference, 10pt]{IEEEtran}
\usepackage[dvips]{graphicx}
\usepackage{amssymb,amsmath}
\usepackage{subfig}
\usepackage{colortbl}

\IEEEoverridecommandlockouts
\newtheorem{theorem}{Theorem}

\newtheorem{corollary}[theorem]{Corollary}
\newtheorem{definition}{Definition}
\newtheorem{example}{Example}

\newcommand{\ie}{\emph{i.e. }}
\newcommand{\etal}{\textit{et al.}}
\newcommand{\m}[1]{\mathcal{#1}}

\begin{document}

\title{Algebraic Network Coding Approach to Deterministic Wireless Relay Networks}
\author{\authorblockN{MinJi Kim, Muriel M\'{e}dard\vspace*{-.2cm}}
\authorblockA{\\ Research Laboratory of Electronics\\
Massachusetts Institute of Technology\\
Cambridge, MA 02139, USA\\
Email: \{minjikim, medard\}@mit.edu}\vspace*{-0.75cm}
}

\maketitle

\begin{abstract}
The deterministic wireless relay network model, introduced by Avestimehr \etal, has been proposed for approximating Gaussian relay networks. This model, known as the ADT network model, takes into account the broadcast nature of wireless medium and interference. Avestimehr \etal \ showed that the Min-cut Max-flow theorem holds in the ADT network.

In this paper, we show that the ADT network model can be described within the algebraic network coding framework introduced by Koetter and M\'{e}dard. We prove that the ADT network problem can be captured by a single matrix, called the \emph{system matrix}. We show that the min-cut of an ADT network is the rank of the system matrix; thus, eliminating the need to optimize over exponential number of cuts between two nodes to compute the min-cut of an ADT network.

We extend the capacity characterization for ADT networks to a more general set of connections. Our algebraic approach not only provides the Min-cut Max-flow theorem for a single unicast/multicast connection, but also extends to non-multicast connections such as multiple multicast, disjoint multicast, and two-level multicast. We also provide sufficiency conditions for achievability in ADT networks for any general connection set. 
In addition, we show that the random linear network coding, a randomized distributed algorithm for network code construction, achieves capacity for the connections listed above. 

Finally, we extend the ADT networks to those with random erasures and cycles (thus, allowing bi-directional links). Note that ADT network was proposed for approximating the wireless networks; however, ADT network is acyclic. Furthermore, ADT network does not model the stochastic nature of the wireless links. With our algebraic framework, we incorporate both cycles as well as random failures into ADT network model.
\end{abstract}
\IEEEpeerreviewmaketitle

\section{Introduction}\label{sec:introduction}

The capacity of the wireless relay networks, unlike its wired counterparts, is still a generally open problem. Even for a simple relay network with one source, one sink, and one relay, the capacity is unknown. In order to better approximate wireless relay networks, \cite{adt1}\cite{adt2} proposed a binary linear deterministic network model (known as the ADT model), which incorporates the broadcast nature of the wireless medium as well as interference. A node within the network receives the bit if the signal is above the noise level; multiple bits that simultaneously arrive at a node are superposed. Note that this model assumes operation under high Signal-to-Noise-Ratio (SNR) -- interference from other users' dominate the noise.

References \cite{adt1}\cite{adt2} characterized the capacity of the ADT networks, and generalized the Min-cut Max-flow theorem for graphs to ADT networks for single unicast/multicast connections. Efficient algorithms to compute the coding strategies to achieve minimum cut has been proposed in \cite{fragouli}\cite{edmund}. Reference \cite{goemans} introduced a flow network, called \emph{linking network}, which generalizes the ADT model, and relates the ADT networks to matroids; thus, allowing the use of matroid theory to solve ADT network problems.

In this paper, we make a connection between the ADT network and algebraic network coding introduced by Koetter and M\'{e}dard \cite{algebraic}, in which they showed an equivalence between the solvability of a network problem and certain algebraic conditions. This paper does not prove or disprove ADT network model's ability to approximate the capacity of the wireless networks, but shows that the ADT network problems, including that of computing the min-cut and constructing a code, can be captured by the algebraic framework.

There are several advantages in generalizing ADT networks to the algebraic network coding formulation. First, this allows the use of results on network coding to better understand the ADT networks. Network coding, proposed in \cite{ahlswede}, allows and encourages \emph{algebraic mixing} of data at intermediate nodes. This mixing maximizes throughput for multicast traffic \cite{ahlswede}, and is robust against failures and erasures \cite{algebraic}. However, most of the classic network coding results consider scalar operations in arbitrary field size, $\mathbb{F}_q$. In order to take advantage of low-complexity operations in $\mathbb{F}_2$, \cite{jaggi} introduces network codes, called \emph{permute-and-add}, that only require bit-wise vector operations, and shows that their performance is still optimal. This shows that network codes in higher field size $\mathbb{F}_q$ can be converted to a binary-vector code without loss in performance.
The connection between ADT networks and algebraic network coding allows the use of existing theorems in the network coding literature to derive new results for ADT networks.


The paper is organized as follows. We present the network model in Section \ref{sec:model}, and an algebraic formulation of the ADT network in Section \ref{sec:algebraic}. Using this algebraic formulation, we provide a definition of min-cut in ADT networks in Section \ref{sec:mincut}. In Sections \ref{sec:singlesource}, we restate the Min-cut Max-flow theorem using our algebraic formulation. In Section \ref{sec:general}, we present new capacity characterizations for ADT networks to a more general set of traffic requirements, such as two-level multicast, disjoint multicast, and multiple source multicast. Note that \cite{adt1}\cite{adt2}\cite{fragouli}\cite{edmund}\cite{goemans} consider single unicast or single multicast connection. Finally, we incorporate random erasures in Section \ref{sec:robust}, and extend the ADT networks to networks with cycles in Section \ref{sec:delay}.

\begin{figure}[tbp]
\begin{center}
\includegraphics[width=0.37\textwidth]{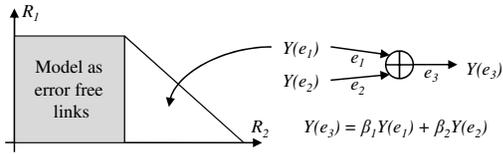}
\end{center}\vspace*{-.3cm}\caption{Additive MAC with two users, and the corresponding rate region. The triangular region is modeled as a set of finite field addition MACs.}\label{fig:mac}\vspace*{-.3cm}
\end{figure}



\section{Network Model}\label{sec:model}

As in \cite{adt1}\cite{adt2}, we shall consider a high SNR regime, in which interference is the dominating factor. In the high SNR regime, the Cover-Wyner region may be well approximated by the combination of two regions, one square and one triangular, as shown in Figure \ref{fig:mac}. The square (shaded) part can be modeled as parallel links for the users, since they do not interfere. The triangular (unshaded) part can be considered as that of a noiseless finite-field addition multiple access channel (MAC) \cite{mac}.
Note in the high SNR regime, analog network coding, which allows and encourages strategic interference, is near optimal \cite{analog_opt}. It is important to note that a network operating in high SNR regime is different from a network with high gain since a large gain amplifies the noise as well as the signal.

The ADT network model uses binary channels, and thus, a binary additive MAC is used to model interference. Prior to \cite{adt1}\cite{adt2}, Effros \etal \ presented an additive MAC over a finite field $\mathbb{F}_q$ \cite{additive}. The Min-cut Max-flow theorem holds for all of the cases above. It may seem that the ADT network model differs greatly from that of \cite{additive} owing to the difference in field sizes used. However, we can achieve a higher field size in ADT networks by combining multiple binary channels and using a \emph{binary-vector} scheme as shown in \cite{jaggi}. In other words, consider two nodes $V_1$ and $V_2$ with two binary channels connecting $V_1$ to $V_2$. Now, instead of considering them as two binary channels, we can ``combine'' the two channels as one with capacity of 2-bits. In this case, instead of using $\mathbb{F}_2$, we can use a larger field size of $\mathbb{F}_4$. Thus, selecting a larger field size $\mathbb{F}_q$, $q >2$ in ADT network model results in fewer but higher capacity parallel channels. Furthermore, it is known that to achieve capacity for multicast connections, $\mathbb{F}_2$ is not sufficient \cite{edmund}; thus, we need to operate in a higher field size. Therefore, in this work, we shall not restrict ourselves to $\mathbb{F}_2$.

\begin{figure}[tbp]
\begin{center}
\includegraphics[width=0.36\textwidth]{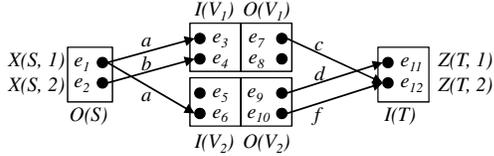}
\end{center}\vspace*{-.3cm}\caption{Example network. We omit $I(S)$ and $O(T)$ in this diagram as they do not participate in the communication.}\label{fig:network}\vspace*{-.3cm}
\end{figure}


We now proceed to defining the network model precisely. A wireless network is modeled using a directed graph $G = (\m{V}, \m{E})$ with a node set $\m{V}$ and an edge set $\m{E}$, as shown in Figure \ref{fig:network}. A node $V\in \m{V}$ consists of \emph{input ports} $I(V)$ and \emph{output ports} $O(V)$. Let $\m{S}, \m{T} \subseteq \m{V}$ be the set of sources and destinations. An edge $(e_1, e_2)$ exists only from an output port $e_1 \in O(V_1)$ to an input port $e_2 \in I(V_2)$, for any $V_1, V_2 \in \m{V}$. Let $\m{E}(V_1, V_2)$ be the set of edges from $O(V_1)$ to $I(V_2)$.
All edges are of unit capacity, where capacity is normalized with respect to the symbol size of $\mathbb{F}_q$. Parallel links of $\m{E}(V_1, V_2)$ deterministically model noise between $V_1$ and $V_2$.

Given such a wireless network $G = (\m{V}, \m{E})$, a source node $S \in \m{S}$ has \emph{independent} random processes $\m{X}(S) = [X(S, 1), X(S, 2), ..., X(S, \mu(S))]$, $\mu(S) \leq |O(S)|$, which it wishes to communicate to a set of destination nodes $\m{T}(S) \subseteq \m{T}$. In other words, we want nodes $T \in \m{T}(S)$ to replicate a subset of the random processes, denoted $\m{X}(S, T) \subseteq \m{X}(S)$, by the means of the network. We define a \emph{connection} $c$ as a triple $(S, T, \m{X}(S, T))$, and the rate of $c$ is defined as $R(c) = \sum_{X(S, i) \in \m{X}(S, T)} H(X(S, i)) = |\m{X}(S,T)|$ (symbols).

\begin{figure}[tbp]
\scriptsize
\begin{align*}
Y(e_1) & = \alpha_{(1, e_1)}X(S, 1) + \alpha_{(2, e_1)}X(S, 2)\\
Y(e_2) & = \alpha_{(1, e_2)}X(S, 1) + \alpha_{(2, e_2)}X(S, 2)\\
Y(e_3) &= Y(e_6) = Y(e_1)\\
Y(e_4) &= Y(e_2)\\
Y(e_5) & = Y(e_{8}) = 0\\
Y(e_{7}) & = \beta_{(e_3, e_{7})}Y(e_3) + \beta_{(e_4, e_{7})}Y(e_4)\\
Y(e_{9}) & = Y(e_{11}) =  \beta_{(e_6, e_{9})}Y(e_6)\\
Y(e_{10}) & = \beta_{(e_6, e_{10})}Y(e_6)\\
Y(e_{12}) & = Y(e_{7}) + Y(e_{10})\\
Z(T, 1) & = \epsilon_{(e_{11}, (T, 1))} Y(e_{11}) +  \epsilon_{(e_{12}, (T, 1))} Y(e_{12})\\
Z(T, 2) & = \epsilon_{(e_{11}, (T, 2))} Y(e_{11}) +  \epsilon_{(e_{12}, (T, 2))} Y(e_{12})
\end{align*}\vspace*{-.3cm}\caption{Equations relating the various processes of Figure \ref{fig:network}.}\label{fig:equations}\vspace*{-.5cm}
\end{figure}

Information is transmitted through the network from the source to the destinations in the following manner. A node $V$ sends information through $e \in O(V)$ at a rate at most one symbol per time unit. Let $Y(e)$ denote the random process at port $e$.
In general, $Y(e)$, $e\in O(V)$, is a function of $Y(e')$, $e'\in I(V)$. In this paper, we consider only linear functions.
\begin{equation}\label{eq:y}
Y(e) = \sum_{e'\in I(V)} \beta_{(e', e)} Y(e'), \text{ for $e \in O(V)$.}
\end{equation}
For a source node $S$, and $e\in O(S)$,
\begin{equation}\label{eq:y-s}
Y(e) =\sum_{e'\in I(V)} \beta_{(e', e)} Y(e') + \sum_{X(S,i) \in \m{X}(S)} \alpha_{(i, e)} X(S, i).
\end{equation}
Finally, the destination $T$ receives a collection of input processes $Y(e')$, $e' \in I(T)$. Node $T$ generates a set of random processes $\m{Z}(T)= [Z(T, 1), Z(T, 2), ..., Z(T, \nu(T))]$ where
\begin{equation}\label{eq:z}
Z(T, i) = \sum_{e' \in I(T)} \epsilon_{(e', (T, i))} Y(e').
\end{equation}
A connection $c = (S, T, \m{X}(S, T))$ is established successfully if $\m{X}(S) = \m{Z}(T)$. A node $V$ is said to \emph{broadcast} to a set $\m{V'}\subseteq \m{V}$ if $\m{E}(V, V') \ne \emptyset$ for all $V'\in  \m{V'}$. In Figure \ref{fig:network}, node $S$ broadcasts to nodes $V_1$ and $V_2$. Superposition occurs at the input port $e' \in I(V)$, \ie $Y(e') = \sum_{(e, e') \in \m{E}} Y(e)$ over a finite field $\mathbb{F}_q$. We say there is a $|\m{V'}|$-user MAC channel if $\m{E}(V', V) \ne \emptyset$ for all $V' \in \m{V'}$. In Figure \ref{fig:network}, nodes $V_1$ and $V_2$ are users, and $T$ the receiver in a 2-user MAC.

For a given network $G$ and a set of connections $\m{C}$, we say that $(G,\m{C})$ is \emph{solvable} if it is possible to establish successfully all connections $c\in \m{C}$. The broadcast and MAC constraints are given by the network; however, we are free to choose the variables $\alpha_{(i, e)}$, $\beta_{(e', e)}$, and $\epsilon_{(e',i)}$ from $\mathbb{F}_q$. Thus, the problem of checking whether a given $(G, \m{C})$ is solvable is equivalent to finding a feasible assignment to $\alpha_{(i,e)}, \beta_{(e', e)}$, and $\epsilon_{(e', (T, i))}$.

\begin{example}\label{ex:equations} The equations in Figure \ref{fig:equations} relate the various processes in the example network in Figure \ref{fig:network}. Note that in Figure \ref{fig:network}, we have set $Y(e_1) =a$, $Y(e_2)=b$, $Y(e_7)=c$, $Y(e_9) = d$, and $Y(e_{10})= f$ for notational simplicity.
\end{example}

\subsection{An Interpretation of the Network Model}\label{sec:interpretation}

\begin{figure}[tbp]
\begin{center}
\includegraphics[width=0.50\textwidth]{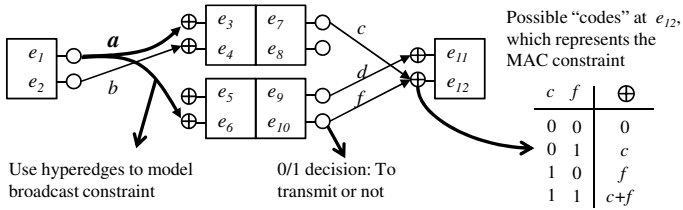}
\end{center}\vspace*{-.3cm}\caption{A new interpretation of the example network from Figure \ref{fig:network}.
}\label{fig:newnetwork}\vspace*{-.5cm}
\end{figure}

The ADT network model uses multiple channels from an output port to model the broadcast channel, and a finite field additive MAC to model interference, as shown in Figure \ref{fig:network}. Note that, in Figure \ref{fig:network}, there are two edges from output port $e_1$ to input ports $e_3$ and $e_6$, respectively; however, due to the broadcast constraint, the two edges $(e_1, e_3)$ and $(e_1, e_6)$ carry the same information $a$. This introduces considerable complexity in constructing a network code as well as computing min-cut of the network \cite{adt1}\cite{adt2}\cite{fragouli}\cite{goemans}. This is due to the fact that the multiple edges from a port do not capture the broadcast dependencies of edges. Furthermore, the broadcast dependencies have to be propagated through the network.

In our approach, we remedy this by introducing the use of hyperedges, as shown in Figure \ref{fig:newnetwork}.
An output port's decision to transmit affects the entire hyperedge; thus, the output port transmits to all the input ports connected to the hyperedge simultaneously. In Section \ref{sec:algebraic}, we shall include the notion of hyperedges in our algebraic formulation to capture the broadcast nature of the wireless medium.
This removes the difficulties of computing the min-cut of ADT networks (Section \ref{sec:mincut}), as it naturally captures the dependencies caused by the broadcasts.

The finite field additive MAC model can be viewed as a set of codes that an input port may receive. As shown in Figure \ref{fig:newnetwork}, input port $e_{12}$ receives one of the four possible codes. The code that $e_{12}$ receives depends on output ports $e_7$'s and $e_9$'s decision to transmit or not.

The difficulty in constructing a network code does not come from any single broadcast or MAC constraint.
The difficulty in constructing a code is in satisfying multiple MAC and broadcast constraints simultaneously. For example, in Figure \ref{fig:constraint}, the fact that $e_4$ may receive $a+b$ does not constrain the choice of $a$ nor $b$. The same argument applies to $e_6$ receiving $a+c$. However, the problem arises from the fact that a choice of value for $a$ at $e_4$ interacts both with $b$ and $c$. As we shall see in Section \ref{sec:mincut}, we eliminate this difficulty by allowing the use of a larger field, $\mathbb{F}_q$.

\begin{figure}[tbp]
\begin{center}
\includegraphics[width=0.18\textwidth]{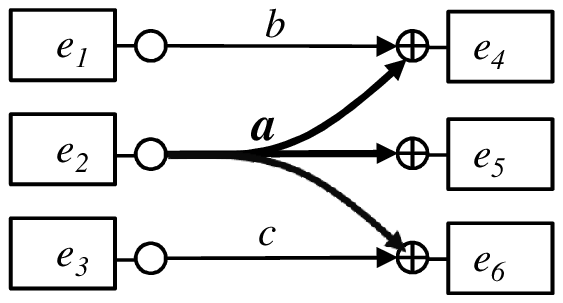}
\end{center}\vspace*{-.3cm}\caption{An example of finite field additive MAC.}\label{fig:constraint}\vspace*{-.3cm}
\end{figure}


\section{Algebraic Formulation}\label{sec:algebraic}

\begin{figure}[tbp]
\begin{center}
\includegraphics[width=0.33\textwidth]{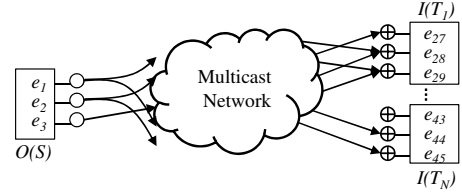}
\end{center}\vspace*{-.3cm}\caption{Single multicast network with source $S$ and receivers $T_1, ..., T_N$.
}\label{fig:multicast}\vspace*{-.5cm}
\end{figure}

We provide an algebraic formulation for the ADT network problem $(G, \m{C})$.
For simplicity, we describe the multicast problem with a single source $S$ and a set of destination nodes $\m{T}$, as in Figure \ref{fig:multicast}. However, this formulation can be extended to multiple source nodes $S_1, S_2, ... S_K$ by adding a super-source $S$ as in Figure \ref{fig:supersource}.

We define a system matrix $M$ to describe the relationship between source's random processes $\m{X}(S)$ and the destinations' processes $\m{Z} = [\m{Z}(T_1), \m{Z}(T_2), ..., \m{Z}(T_{|\m{T}|})]$. Thus, we want to characterize $M$ where
\begin{equation}
\m{Z} = \m{X}(S) \cdot M.
\end{equation}
The matrix $M$ is composed of three matrices, $A$, $F$, and $B$.


Given $G$, we define the adjacency matrix $F$ as follows:
\begin{equation}\label{eq:F}F_{i, j}= \begin{cases}
1 & \text{if $(e_i, e_j) \in \m{E}$,}\\
\beta_{(e_i, e_j)} & \text{if $e_i \in I(V)$, $e_j \in O(V)$ for $V \in \m{V}$,}\\
0 & \text{otherwise.}
\end{cases}
\end{equation}

Matrix $F$ is defined on the ports, rather than on the nodes. This is because, in the ADT model, each port is the basic receiver/transmitter unit. Each entry $F_{i,j}$ represents the input-output relationships of the ports. A zero entry indicates that the ports are not directly connected, while an entry of one represents that they are connected. The adjacency matrix $F$ naturally captures the physical structure of the ADT network. Note that a row with multiple entries of 1 represent the broadcast hyperedge; while a column with multiple entries of 1 represent the MAC constraint. Note that the 0-1 entries of $F$ represent the \emph{fixed} network topology as well as the broadcast and MAC constraints. On the other hand, $\beta_{(e_i, e_j)}$ are free variables, representing the coding coefficients used at $V$ to map the input port processes to the output port processes. This is the key difference between the work presented here and in \cite{algebraic} -- $F$ is partially fixed in the ADT network model due to network topology and broadcast/MAC constraints, while in \cite{algebraic}, only the network topology affect $F$.

In \cite{adt1}\cite{adt2}, the nodes are allowed to perform any internal operations; while in \cite{fragouli}\cite{goemans}, only permutation matrices (\ie routing) are allowed. In their work \cite{adt1}\cite{adt2}, the authors also show that linear operations are sufficient for achieving capacity in ADT networks for a single multicast traffic. We propose a general setup in which $\beta_{(e_i, e_j)} \in \mathbb{F}_q$ -- thus, allowing any matrix operation, as in \cite{adt1}\cite{adt2}.



\begin{figure}[tbp]
\begin{center}
\includegraphics[width=.47\textwidth]{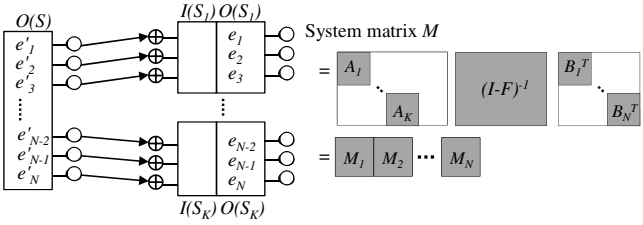}
\end{center}\vspace*{-.4cm}\caption{A network with multiple sources $S_1, S_2, ..., S_K$ can be converted to a single source problem by adding a super-source $S$ with $|O(S)| = \sum_{i=1}^K |O(S_i)|$. Each $e'_j \in O(S)$ has a ``one-to-one connection'' to a $e_j \in O(S_i)$, for $i \in [1, K]$. Matrix $A_i$ represent the encoding matrix for source $S_i$, while $B_j$ is the decoding matrix at destination $T_j$. The white area represents the zero elements, and the shaded area represents the coding coefficients.}\label{fig:supersource}\vspace*{-.5cm}
\end{figure}

Note that $F^k$, the $k$-th power of an adjacency matrix of a graph $G$, shows the existence of paths of length $k$ between any two nodes in $G$. Therefore, the series $I + F + F^2 + F^3 + ...$ represent the connectivity of the network. It can be verified that $F$ is nilpotent, which means that there exists a $k$ such that $F^k$ is a zero matrix. As a result, $I + F + F^2 + F^3 + ...$ can be written as $(I-F)^{-1}$. Thus, $(I-F)^{-1}$ represent the impulse response of the network. Note that, $(I-F)^{-1}$ exists for all acyclic network since $I-F$ is an upper-triangle matrix with all diagonal entries equal to 1; thus, $\det(I-F) = 1$.

\begin{example}\label{ex:F}
In Figure \ref{fig:F}, we provide the $12 \times 12$ adjacency matrix $F$ for the example network in Figures \ref{fig:network} and \ref{fig:newnetwork}. Note that the first row (with two entries of 1) represents the broadcast hyperedge, $e_1$ connected to both $e_3$ and $e_6$. The last column with two entries equal to 1 represents the MAC constraint, both $e_7$ and $e_{10}$ talking to $e_{12}$.
The highlighted elements in $F$ represent the coding variables, $\beta_{(e', e)}$, of $V_1$ and $V_2$ in Figure \ref{fig:newnetwork}. For some $(e', e)$, $\beta_{(e',e)} =0$ since these ports of $V_1$ and $V_2$ are not used.
\end{example}


Matrix $A$ represents the encoding operations performed at $S$. We define a $|\m{X}(S)| \times |\m{E}|$ encoding matrix $A$ as follows:
\begin{equation}
A_{i,j} = \begin{cases}
\alpha_{(i, e_j)}  & \text{if $e_j\in O(S)$ and $X(S,i) \in \m{X}(S)$,}\\
0& \text{otherwise}.
\end{cases}
\end{equation}

\begin{example}\label{ex:A}
We provide the $2 \times 12$ encoding matrix $A$ for the network in Figure \ref{fig:network}.
\[
A = \begin{pmatrix}
\alpha_{1, e_1} & \alpha_{1, e_2} & 0 & \dotsb & 0\\
\alpha_{2, e_1} & \alpha_{2, e_2} & 0 & \dotsb & 0\\
\end{pmatrix}.
\]
\end{example}


Matrix $B$ represents the decoding operations performed at the destination nodes $T \in \m{T}$. Since there are $|\m{T}|$ destination nodes, $B$ is a matrix of size $|\m{Z}| \times |\m{E}|$ where $\m{Z}$ is the set of random processes derived at the destination nodes. We define the decoding matrix $B$ as follows:
\begin{equation}
B_{i,(T_j, k)} = \begin{cases}
\epsilon_{(e_i, (T_j, k))} & \text{if $e_i \in I(T_j), Z(T_j, k) \in \m{Z}(T_j)$},\\
0 & \text{otherwise.}
\end{cases}
\end{equation}

\begin{example}\label{ex:B}
We provide the $2 \times 12$ decoding matrix $B$ for the example network in Figure \ref{fig:network}.
\[
B = \begin{pmatrix}
0 & \dotsb & 0 & \epsilon_{(e_{11}, (T, 1))}  & \epsilon_{(e_{12}, (T, 1))} \\
0 & \dotsb & 0 & \epsilon_{(e_{11}, (T, 2))}  & \epsilon_{(e_{12}, (T, 2))} \\
\end{pmatrix}.
\]
\end{example}
\begin{figure}[tbp]
\[\scriptsize
\left(
\begin{array}{cccccccccccc}
0 & 0 & 1 & 0 & 0 & 1 & 0 & 0 & 0 & 0 & 0 & 0\\
0 & 0 & 0 & 1 & 0 & 0 & 0 & 0 & 0 & 0 & 0 & 0\\
0 & 0 & 0 & 0 & 0 & 0 & \cellcolor[gray]{.8} \beta_{(e_3, e_7)} & \cellcolor[gray]{.8} 0 & 0 & 0 & 0 & 0\\
0 & 0 & 0 & 0 & 0 & 0 & \cellcolor[gray]{.8} \beta_{(e_4, e_7)} & \cellcolor[gray]{.8} 0 & 0 & 0 & 0 & 0\\
0 & 0 & 0 & 0 & 0 & 0 & 0 & 0 & \cellcolor[gray]{.8} 0 & \cellcolor[gray]{.8} 0 & 0 & 0\\
0 & 0 & 0 & 0 & 0 & 0 & 0 & 0 & \cellcolor[gray]{.8} \beta_{(e_6, e_{9})} & \cellcolor[gray]{.8} \beta_{(e_6, e_{10})}& 0 & 0\\
0 & 0 & 0 & 0 & 0 & 0 & 0 & 0 & 0 & 0 & 0 & 1\\
0 & 0 & 0 & 0 & 0 & 0 & 0 & 0 & 0 & 0 & 0 & 0\\
0 & 0 & 0 & 0 & 0 & 0 & 0 & 0 & 0 & 0 & 1 & 0\\
0 & 0 & 0 & 0 & 0 & 0 & 0 & 0 & 0 & 0 & 0 & 1\\
0 & 0 & 0 & 0 & 0 & 0 & 0 & 0 & 0 & 0 & 0 & 0\\
0 & 0 & 0 & 0 & 0 & 0 & 0 & 0 & 0 & 0 & 0 & 0\\
\end{array}
\right)
\]\vspace*{-.3cm}\caption{$12 \times 12$ adjacency matrix $F$ for network in Figure \ref{fig:network}.}\vspace*{-.1cm}\label{fig:F}
\end{figure}

\begin{theorem}\label{thm:m}
Given a network $G = (\m{V}, \m{E})$, let $A$, $B$, and $F$ be the encoding, decoding, and adjacency matrices, respectively. Then, the system matrix $M$ is given by
\begin{equation}
M = A (1-F)^{-1} B^T.
\end{equation}
\end{theorem}
\begin{proof}
The proof of this theorem is similar to that of Theorem 3 in \cite{algebraic}. As previously mentioned, $(I-F)^{-1} = (I + F + F^2 + ...)$ always exists for an acyclic network $G$.
\end{proof}

Note that the algebraic framework shows a clear separation between the given physical constraints (fixed 0-1 entries of $F$ showing the topology and the broadcast/MAC constraints), and the coding decisions. As mentioned previously, we can freely choose the coding variables $\alpha_{(i, e_j)}$, $\epsilon_{(e_i, (T_j, k))}$, and $\beta_{(e_i, e_j)}$. Thus, solvability of $(G, \m{C})$ is equivalent to assigning values to $\alpha_{(i, e_j)}$, $\epsilon_{(e_i, (T_j, k))}$, and $\beta_{(e_i, e_j)}$ such that each receiver $T \in \m{T}$ is able to decode the data it is intended to receive.

\begin{example}
We can combine the matrices $F$, $A$, and $B$ from Examples \ref{ex:F}, \ref{ex:A}, and \ref{ex:B} respectively to obtain the system matrix $M = A(I-F)^{-1}B^T$ for the network in Figure \ref{fig:network}. We show a schematic of the system matrix $M$ in Figure \ref{fig:systemmatrix}.
\end{example}

\begin{figure}[tbp]
\begin{center}
\includegraphics[width=0.4\textwidth]{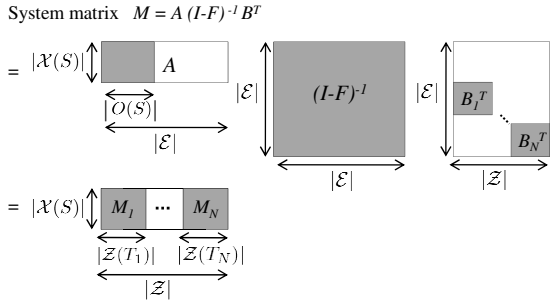}\vspace*{-.15cm}
\caption{The system matrix $M$ and it's components $A$, $(I-F)^{-1}$, and $B$ for a single multicast connection with source $S$ and destinations $T_i$, $i\in [1, N]$. 
}\label{fig:systemmatrix}\vspace*{-.5cm}
\end{center}
\end{figure}

\section{Definition of Min-cut}\label{sec:mincut}

Consider a source $S$ and a destination $T$. Reference \cite{adt1} proves the maximal achievable rate to be the minimum value of all $S$-$T$ cuts, denoted $mincut(S,T)$, which we reproduce below in Definition \ref{def:mincut}.

\begin{definition}\label{def:mincut}\cite{adt1}\cite{adt2} A cut $\Omega$ between a source $S$ and a destination $T$ is a partition of the vertices into two disjoint sets $\Omega$ and $\Omega^c$ such that $S \in \Omega$ and $T \in \Omega^c$. For any cut, $G_{\Omega}$ is the incidence matrix associated with the bipartite graph with ports in $\Omega$ and $\Omega^c$. Then, the capacity of the given ADT network (equivalently, $mincut(S,T)$) is defined as
\[
mincut(S,T) = \min_{\Omega} rank(G_{\Omega}).
\]
This capacity of $mincut(S,T)$ can be achieved using linear operations for a single unicast/multicast connection.\hspace*{1cm} $\blacksquare$
\end{definition}

Note that, with the above definition, in order to compute $mincut(S,T)$, we need to optimize over all cuts between $S$ and $T$. In addition, the proof of achievability in \cite{adt1} is not constructive, as it assumes infinite block length and does not consider the details of internal node operations.

We introduce a new algebraic definition of the min-cut, and show that it is equivalent to that of Definition \ref{def:mincut}.

\begin{theorem}\label{thm:mincut}
The capacity of the given ADT, equivalently the minimum value of all $S-T$ cuts $mincut(S,T)$, is
\begin{align*}
mincut(S,T) &= \min_{\Omega} \text{rank}(G_{\Omega})\\
&= \max_{\alpha_{(i, e)}, \beta_{(e', e)},\epsilon_{(e',i)}} \text{rank}(M).
\end{align*}
\end{theorem}
\begin{proof}
By \cite{adt1}, we know that $mincut(S,T) = \min_{\Omega} \text{rank}(G_{\Omega})$. Therefore, we show that $\max_{\alpha, \beta,\epsilon} \text{rank}(M)$ is equivalent to the maximal achievable rate in an ADT network.

First, we show that $mincut(S,T) \geq  \max_{\alpha, \beta,\epsilon} \text{rank}(M)$. In our algebraic formulation, $\m{Z}(T) = \m{X}(S)M$; thus, the rank of $M$ represents the rate achieved. Let $R = \max_{\alpha, \beta,\epsilon} \text{rank}(M)$. Then, there exists an assignment of $\alpha_{(i, e)}, \beta_{(e', e)},$ and $\epsilon_{(e',i)}$ such that the network achieves a rate of $R$. By the definition of min-cut, it must be the case that $mincut(S,T) \geq R$.

Second, we show that $mincut(S,T) \leq  \max_{\alpha, \beta,\epsilon} \text{rank}(M)$. Assume that $R = mincut(S,T)$. Then, by \cite{adt1}\cite{adt2}, there exists a linear configuration of the network such that we can achieve a rate of $R$ such that the destination node $T$ is able to reproduce $\m{X}(S,T)$. This configuration of the network provides a linear relationship of the source-destination processes (actually, the resulting system matrix is an identity matrix); thus, an assignment of the variables $\alpha_{(i, e)}, \beta_{(e', e)}$, and $\epsilon_{(e',i)}$ for our algebraic framework. We denote $M'$ to be the system matrix corresponding to this assignment. Note that, by the definition, $M'$ is an $R\times R$ matrix with a rank of $R$. Therefore, $\max_{\alpha, \beta,\epsilon} \text{rank}(M) \geq \text{rank}(M') = mincut(S,T)$.
\end{proof}

The system matrix $M$ 
depends not only on the structure of the ADT network, but also on the field size used, nodes' internal operations, transmission rate, and connectivity. For example, the network topology may change with a choice of larger field size, since larger field sizes result in fewer parallel edges/channels.
Another example, if we adjust the rate such that $|\m{X}(S)| \leq mincut(S,T)$, then $M$ has full-rank. However, if $|\m{X}(S)|> mincut(S,T)$, then $M$ may have rank of $mincut(S,T)$ but not be full-rank. It is important to note that, in ADT networks, the cut value may not equal to the graph theoretical cut value (see Figure 2 in \cite{fragouli}).

\section{Min-cut Max-flow Theorem}\label{sec:singlesource}

In this section, we provide an algebraic interpretation of the Min-cut Max-flow theorem for a single unicast connection and a single multicast connection \cite{adt1}\cite{adt2}. This result is a direct consequence of \cite{algebraic} when applied to the algebraic formulation for the ADT network. In addition, we show that a distributed randomized coding scheme achieves capacity for these connections.

\begin{theorem}[Min-cut Max-flow Theorem] Given an acyclic network $G$ with a single connection $c = (S, T, \m{X}(S, T))$ of rate $R(c) = |\m{X}(S, T)|$, the following are equivalent.
\begin{enumerate}
\item A unicast connection $c$ is feasible.
\item $mincut(S,T) \geq R(c)$.
\item There exists an assignment of $\alpha_{(i, e_j)}$, $\epsilon_{(e_i, (T_j, k))}$, and $\beta_{(e_i, e_j)}$ such that the $R(c) \times R(c)$ system matrix $M$ is invertible in $\mathbb{F}_q$ (\ie $\det(M) \ne 0$).
\end{enumerate}\label{thm:mincut_maxflow}
\end{theorem}

\begin{proof}
Statements 1) and 2) have been shown to be equivalent in ADT network models \cite{adt1}\cite{fragouli}\cite{goemans}. From Theorem \ref{thm:mincut}, we have shown the equivalence between $mincut(S,T) =  \max_{\alpha, \beta,\epsilon} \text{rank}(M)$. Therefore, for any rate $R(c) \leq mincut(S,T)$, $M$ is a full-rank square matrix. Thus, $M$ is invertible.
\end{proof}

\begin{corollary}[Random Coding for Unicast]Consider an ADT network problem with a single connection $c = (S, T, \m{X}(S, T))$ of rate $R(c) = |\m{X}(S, T)| \leq mincut(S,T)$. Then, random linear network coding, where some or all code  variables $\alpha_{(i, e_j)}$, $\epsilon_{(e_i, (T_j, k))}$, and $\beta_{(e_i, e_j)}$ are chosen independently and uniformly over all elements of $\mathbb{F}_q$, guarantees decodability at destination node $T$ with high probability at least $(1-\frac{1}{q})^{\eta}$, where $\eta$ is the number of links carrying random combinations of the source processes.
\end{corollary}
\begin{proof}From Theorem \ref{thm:mincut_maxflow}, there exists an assignment of $\alpha_{(i, e_j)}$, $\epsilon_{(e_i, (T_j, k))}$, and $\beta_{(e_i, e_j)}$ such that $\det(M)\ne 0$, which gives a capacity-achieving network code for the given $(G, \m{C})$. Thus, this connection $c$ is feasible for the given network. Reference \cite{rlc} proves that random linear network coding is capacity-achieving and guarantees decodability with high probability $(1-\frac{1}{q})^{\eta}$ for such feasible unicast connection $c$.
\end{proof}

%

\begin{theorem}[Single Multicast Theorem] Given an acyclic network $G$ and connections $\m{C}= \{(S, T_1, \m{X}(S)),$ $(S, T_2,$ $\m{X}(S)),$ $..., (S, T_N, \m{X}(S))\}$, $(G, \m{C})$ is solvable if and only if $mincut(S, T_i) \geq |\m{X}(S)|$ for all $i$. \label{thm:singlemulticast}
\end{theorem}

\begin{proof} If $(G, \m{C})$ is solvable, then $mincut(S, T_i) \geq |\m{X}(S)|$. Therefore, we only have to show the converse. Assume $mincut(S, T_i) \geq |\m{X}(S)|$ for all $i \in [1, N]$. The system matrix $M = \{M_i\}$ is a concatenation of $|\m{X}(S)| \times |\m{X}(S)|$ matrix where $\m{Z}(T_i) = \m{X}(S) M_i$, as shown in Figure \ref{fig:systemmatrix}. We can write $M = [M_1, M_2, ..., M_N] = A (I-F)^{-1} B^T = A(I-F)^{-1}[B_1, B_2, ..., B_N]$. Thus, $M_i = A(I-F)^{-1}B_i$. Note that $A$ and $B_i$'s do not substantially contribute to the system matrix $M_i$ since $A$ and $B_i$ only perform linear encoding and decoding at the source and destinations, respectively.

By Theorem \ref{thm:mincut_maxflow}, there exists an assignment of $\alpha_{(i, e_j)}$, $\epsilon_{(e_i, (T_j, k))}$, and $\beta_{(e_i, e_j)}$ such that each individual system submatrix $M_i$ is invertible, \ie $\det{(M_i)} \ne 0$. However, an assignment that makes $\det{(M_i)} \ne 0$ may lead to $\det{(M_j)} = 0$ for $i\ne j$. Thus, we need to show that it is possible to achieve \emph{simultaneously} $\det{(M_i)} \ne 0$ for all $i$. By \cite{rlc}, we know that if the field size is larger than the number of receivers ($q > N$), then there exists an assignment of $\alpha_{(i, e_j)}$, $\epsilon_{(e_i, (T_j, k))}$, and $\beta_{(e_i, e_j)}$ such that $\det{(M_i)} \ne 0$ for all $i$.
\end{proof}

\begin{corollary}[Random Coding for Multicast] Consider an ADT network problem with a single multicast connection $\m{C}= \{(S,T_1, \m{X}(S)), (S, T_2, \m{X}(S)), ..., (S, T_N, \m{X}(S))\}$ with $mincut(S,T_i) \geq |\m{X}(S)|$ for all $i$.  Then, random linear network coding, where some or all code variables $\alpha_{(i, e_j)}$, $\epsilon_{(e_i, (T_j, k))}$, and $\beta_{(e_i, e_j)}$ are chosen independently and uniformly over all elements of $\mathbb{F}_q$, guarantees decodability at destination node $T_i$ for all $i$ simultaneously with high probability at least $(1-\frac{N}{q})^{\eta}$, where $\eta$ is the number of links carrying random combinations of the source processes; thus, $\eta \leq |\m{E}|$.\label{thm:coding_multicast}
\end{corollary}
\begin{proof}Given that the multicast connection is feasible (which is true by Theorem \ref{thm:singlemulticast}), reference \cite{rlc} shows that random linear network coding achieves capacity for multicast connections, and allows all destination nodes to decode the source processes $\m{X}(S)$ with high probability $(1-\frac{N}{q})^{\eta}$.
\end{proof}

\section{Extensions to other connections}\label{sec:general}

In this section, we extend the ADT network results to a more general set of traffic requirements. We use the algebraic formulation and the results from \cite{algebraic} to characterize the feasibility conditions for a given problem $(G, \m{C})$.

\begin{theorem}[Multiple Multicast Theorem] Given a network $G$ and a set of connections $\m{C} = \{(S_i, T_j, \m{X}(S_i))\ |\ S_i \in \m{S}, T_j \in \m{T}\}$, $(G, \m{C})$ is solvable if and only if Min-cut Max-flow bound is satisfied for any cut between source nodes $\m{S}$ and a destination $T_j$, for all $T_j \in \m{T}$.\label{thm:multiplemulticast}
\end{theorem}
\begin{proof} We first introduce a super-source $S$ with $|O(S)| = \sum_{S_i \in \m{S}} |O(S_i)|$, and connect each $e'_j \in O(S)$ to an input of $S_i$ such that $e_j \in O(S_i)$ as shown in Figure \ref{fig:supersource}. Then, we apply Theorem \ref{thm:singlemulticast}, which proves the statement.
\end{proof}

\begin{theorem}[Disjoint Multicast Theorem] Given an acyclic network $G$ with a set of connections $\m{C} = $ $ \{(S, T_i, $ $ \m{X}(S, T_i))$ $\ | \ i = 1,2, ...,K\}$ is called a \emph{disjoint multicast} if $\m{X}(S, T_i) \cap \m{X}(S, T_j) = \emptyset$ for all $i\ne j$. Then, $(G, \m{C})$ is solvable if and only if $mincut(S, \m{T}') \geq \sum_{T_i \in \m{T'}} |\m{X}(S, T_i)|$ for any $\m{T}' \subseteq \m{T}$.\label{thm:disjoint}
\end{theorem}
\begin{proof}
Create a super-destination node $T$ with $|I(T)| = \sum_{i=1}^K |I(T_i)|$, and an edge $(e, e')$ from $e\in O(T_i)$, $i\in [1, K]$ to $e'\in I(T)$, as in Figure \ref{fig:superdestin}. This converts the problem of disjoint multicast to a single-source $S$, single-destination $T$ problem with rate $\m{X}(S, T) = \sum_{T' \in \m{T}} |\m{X}(S, T)|$. The $mincut(S, T) \geq |\m{X}(S, T)|$; so, Theorem \ref{thm:mincut_maxflow} applies. Thus, it is possible to achieve a communication of rate $\m{X}(S,T)$ between $S$ and $T$. Now, we have to guarantee that the receiver $T_i$ is able to receive the exact subset of processes $\m{X}(S,T_i)$. Since the system matrix to $T$ is full rank, it is possible to carefully choose the encoding matrix $A$ such that the system matrix $M$ at super-destination node $T$ is an identity matrix. This implies that for each edge from the output ports of $T_i$ (for all $i$) to input ports of $T$ is carrying a distinct symbol, disjoint from all the other symbols carried by those  edges from output ports of $T_j$, for all $i \ne j$. Thus, by appropriately permuting the symbols at the source, $S$ can deliver the desired processes to the intended $T_i$ as shown in Figure \ref{fig:superdestin}. 
\end{proof}

\begin{figure}
\begin{center}
\includegraphics[width=0.45\textwidth]{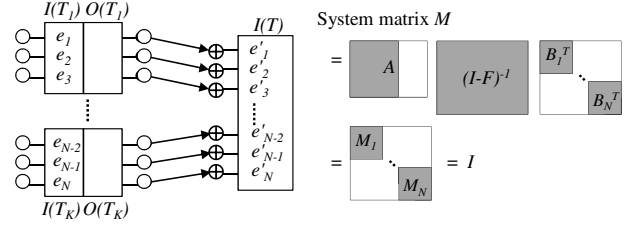}
\end{center}\vspace*{-.3cm}\caption{Disjoint multicast problem can be converted into a single destination problem by adding a super-destination $T$. The system matrix $M$ for the disjoint multicast problem is shown as well.
}\label{fig:superdestin}\vspace*{-.2cm}
\end{figure}

\begin{theorem}[Two-level Multicast Theorem] Given an acyclic network $G$ with a set of connections $\m{C} =\m{C}_{d} \cup \m{C}_{m}$ where $\m{C}_{d} = \{(S, T_i, \m{X}(S, T_i)) | \m{X}(S, T_i) \cap \m{X}(S, T_j) = \emptyset,$ $ i\ne j$, $i, j \in [1,K]\}$ is a set of disjoint multicast connections, and $\m{C}_{m} = \{(S, T_i, \m{X}(S))\ | \ i \in [K+1, N]\}$ is a set of single source multicast connections. Then, $(G, \m{C})$ is solvable if and only if the min-cut between $S$ and any $\m{T}' \subseteq \{T_1, ..., T_K\}$ is at least $\sum_{T_i \in \m{T}'} |\m{X}(S, T_i)|$, and min-cut between $S$ and $T_j$ is at least $|\m{X}(S)|$ for $j \in [K+1, N]$.\label{thm:twolevel}
\end{theorem}
\begin{proof}
We create a super-destination $T$ for the disjoint multicast destinations as in the proof for Theorem \ref{thm:disjoint}. Then, we have a single multicast problem with receivers $T$ and $T_i$, $i\in [K+1, N]$. Theorem \ref{thm:singlemulticast} applies. By choosing the appropriate matrix $A$, $S$ can satisfy both the disjoint multicast and the single multicast requirements, as shown in Figure \ref{fig:twolevel}.
\end{proof}

Theorem \ref{thm:twolevel} does not extend to a three-level multicast.

\begin{figure}[tbp]
\begin{center}
\includegraphics[width=0.47\textwidth]{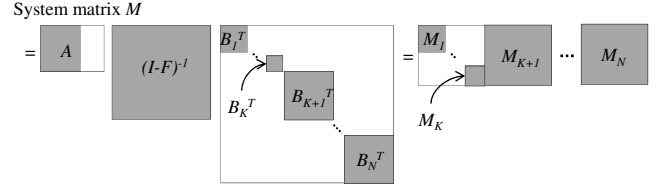}
\end{center}\vspace*{-.3cm}\caption{The system matrix $M$ for the two-level multicast problem. The structure of the system matrix $M$ is a ``concatenation'' of the disjoint multicast problem (Figure \ref{fig:superdestin}) and the single multicast problem (Figure \ref{fig:multicast}).}\label{fig:twolevel}\vspace*{-.4cm}
\end{figure}

In the theorem below, we present sufficient conditions for solvability of a general connection set. This theorem does not provide necessary conditions, as shown in \cite{insufficiency}.

\begin{theorem}[Generalized Min-cut Max-flow Theorem] Given an acyclic network $G$ with a connection set $\m{C}$, let $M = \{M_{i,j}\}$ where $M_{i,j}$ is the system matrix for source processes $\m{X}(S_i)$ to destination processes $\m{Z}(T_j)$. Then, $(G, \m{C})$ is solvable if there exists an assignment of $\alpha_{(i, e_j)}$, $\epsilon_{(e_i, (T_j, k))}$, and $\beta_{(e_i, e_j)}$ such that
\begin{enumerate}
\item $M_{i,j} = 0$ for all $(S_i, T_j, \m{X}(S_i, T_j)) \notin \m{C}$,
\item Let $(S_{\sigma(i)}, T_j, \m{X}(S_{\sigma(i)}, T_j)) \in \m{C}$ for $i \in [1, K(j)]$. Thus, this is the set of connections with $T_j$ as a receiver. Then, $[M_{\sigma(1),j}^T, M_{\sigma(2), j}^T, ...,$ $M_{\sigma(K_j), j}^T]$ is a $|\m{Z}(T_j))| \times |\m{Z}(T_j)|$ is a \emph{nonsingular}
    system matrix.
\end{enumerate}
\end{theorem}
\begin{proof}
Note that $[M_{\sigma(1),j}^T, M_{\sigma(2), j}^T, ...,$ $M_{\sigma(K_j), j}^T]$  is a system matrix for source processes $\m{X}(S_{\sigma(i)})$, $i \in [1, K(j)]$, to destination processes $\m{Z}(T_j)$.

Condition 2) states the Min-cut Max-flow condition; thus, is necessary to establish the connections. Condition 1) states that the destination node $T_j$ should be able to distinguish the information it is intended to receive from the information that may have been mixed into the flow it receives. These two conditions are sufficient to establish all connections in $\m{C}$. We do not provide the details for want of space; however, the proof is similar to that of Theorem 6 in \cite{algebraic}.
\end{proof}

We briefly note the capacity achieving code construction for the non-multicast connections described in this section. For multiple multicast, a random linear network coding approach achieves capacity -- \ie the source nodes and the intermediate nodes can randomly and uniformly select coding coefficients. However, a minor modification is necessary for disjoint multicast and two-level multicast. We note that only the source's encoding matrix $A$ needs to be modified. As in the proofs of Theorems \ref{thm:disjoint} and \ref{thm:twolevel}, the intermediate nodes can randomly and uniformly select coding coefficients; thus, preserving the distributed and randomized aspect of the code construction. Once the coding coefficients at the intermediate nodes are selected, $S$ carefully chooses the encoding matrix $A$ such that the system matrix corresponding to the receivers of the disjoint multicast (in the two-level multicast, these would correspond to $T_i$, $i \in [1,K]$) is an identity matrix. This can be done because the system matrix $M$ is full rank.

\section{Network with Random Erasures}\label{sec:robust}

We consider the algebraic ADT problem where links may fail randomly, and cause erasures. Wireless networks are stochastic in nature, and random erasures occur dynamically over time. However, the original ADT network models noise deterministically with parallel noise-free bit-pipes. As a result, the min-cut (Definition \ref{def:mincut}) and the network code \cite{fragouli}\cite{edmund}\cite{goemans}, which depend on the hard-coded representation of noise, have to be recomputed every time the network changes.

We show that the algebraic framework for the ADT network is robust against random erasures and failures. First, we show that for some set of link failures, the network code remain successful. This translate to whether the system matrix $M$ preserves its full rank even after a subset of variables $\alpha_{(i, e_j)}, \epsilon_{(e_i, (D_j, k))}$, and $\beta_{(e_i, e_j)}$ associated with the failed links is set to zero. Second, we show that the specific instance of the system matrix $M$ and its rank are not as important as the \emph{average} $\text{rank}(M)$ when computing the time average min-cut. Note that the original min-cut definition (Definition \ref{def:mincut}) requires an optimization over exponential number of cuts for every time step to find the time average min-cut. With this insight, we shall use the results from \cite{reliable} to show that random linear network coding achieves the time-average min-cut, \ie is capacity-achieving.

We assume that any link within the network may fail. Given an ADT network $G$ and a set of link failures $f$, $G_f$ represents the network $G$ experiencing failures $f$. This can be achieved by deleting the failing links from $G$, which is equivalent to setting the coding variables in $B(f)$ to zero, where $B(f)$ is the set of coding variables associated with the failing links. We denote $M$ be the system matrix for network $G$. Let $M_{f}$ be the system matrix for the network $G_f$.

\subsection{Robust against Random Erasures}\label{sec:erasures}

Given an ADT network problem $(G, \m{C})$, let $\m{F}$ be the set of \emph{all} link failures such that, for any $f\in \m{F}$, the problem $(G_f, \m{C})$ is solvable. The solvability of a given $(G_f, \m{C})$ can be verified using resulting in Sections \ref{sec:singlesource} and \ref{sec:general}.  We are interested in static solutions, where the network is oblivious of $f$. In other words, we are interested in finding the set of link failures such that the network code is still successful in delivering the source processes to the destinations. For a multicast connection, we show the following surprising result.

\begin{theorem}[Static Solution for Random Erasures]\label{thm:static}
Given an ADT network problem $(G, \m{C})$ with a multicast connection $\m{C} = \{(S, T_1, \m{X}(S)),$ $(S, T_2, \m{X}(S)), ..., (S, T_N, \m{X}(S))\}$, there exists a \emph{static} solution to the problem $(G_f, \m{C})$ for all $f\in \m{F}$.
\end{theorem}
\begin{proof}
By Theorem \ref{thm:singlemulticast}, we know that for any given $f\in \m{F}$, the problem $(G_f, \m{C})$ is solvable; thus, there exists a code $\det{(M_f)}\ne 0$. Now, we need to show that there exists a code such that $\det{(M_f)}\ne 0$ for all $f\in \m{F}$ simultaneously. This is equivalent to finding a non-zero solution to the following polynomial: $\prod_{f\in \m{F}} \det{(M_f)} \ne 0$. Reference \cite{rlc} showed that if the field size is large enough ($q > |\m{F}||\m{T}| = |\m{F}|N$), then there exists an assignment of $\alpha_{(i, e_j)}, \epsilon_{(e_i, (D_j, k))}$, and $\beta_{(e_i, e_j)}$ such that $\det{(M_f)} \ne 0$ for all $ f \in \m{F}$.
\end{proof}



\begin{corollary}[Random Coding against Random Erasures]\label{thm:static_coding} Consider an ADT network problem with a multicast connection $\m{C} = \{(S, T_1, \m{X}(S)),$ $(S, T_2, \m{X}(S)), ..., (S, T_N, \m{X}(S))\}$, which is solvable under link failures $f$, for all $f\in \m{F}$. Then, random linear network coding, where some or all code variables $\alpha_{(i, e_j)}, \epsilon_{(e_i, (D_j, k))}$, and $\beta_{(e_i, e_j)}$ are chosen independently and uniformly over all elements of $\mathbb{F}_q$ guarantees decodability at destination nodes $T_i$ for all $i$ simultaneously and remains successful regardless of the failure pattern $f \in \m{F}$ with high probability at least $(1- \frac{N|\m{F}|}{q})^\eta$, where $\eta$ is the number of links carrying random combinations of the source processes.
\end{corollary}
\begin{proof}
Given a multicast connection that is feasible under any link failures $f \in \m{F}$, reference \cite{rlc} shows that random linear network coding achieves capacity for multicast connections, and is robust against any link failures $f \in \m{F}$ with high probability $(1 - \frac{N |\m{F}|}{q})^\eta$.
\end{proof}

We note that it is unclear whether this can be extended to the non-multicast connections, as noted in \cite{algebraic}. Reference \cite{algebraic} shows a simple example network in which no static solution is available for a set of feasible failure patterns.

\subsection{Time-average Min-cut}\label{sec:average}

In this section, we study the time-average behavior of the ADT network, given random erasures. We use techniques from \cite{reliable}, which studies reliable communication over lossy networks with network coding.

Consider an ADT network $G$. In order to study time-average behavior, we introduce erasure distributions. Let $\m{F}'$ be a set of link failure patterns in $G$. Assume that any set of link failures $f \in \m{F'}$ may occur with probability $p_f$. In this section, we study the average behavior of the network over a long period of time; thus, the steady state behavior.

\begin{theorem}[Min-cut for Time-varying Network]\label{thm:average_mincut}
Assume an ADT network $G$ in which link failure pattern $f \in \m{F}'$ occurs with probability $p_f$. Then, the average min-cut between two nodes $S$ and $T$ in $G$, $mincut_{\m{F}'}(S,T)$ is
\[
mincut_{\m{F}'} (S,T) = \sum_{f\in \m{F}'} p_f \left(\max_{\alpha_{(i, e)}, \beta_{(e', e)}, \epsilon_{(e', i)}} \text{rank}(M_f)\right).
\]
\end{theorem}
\begin{proof}
By Theorem \ref{thm:mincut}, we know that at any given time instance with failure pattern $f$, the min-cut between $S$ and $T$ is given by $\max_{\alpha_{(i, e)}, \beta_{(e', e)}, \epsilon_{(e', i)}} \text{rank}(M_f)$. Then, the above statement follows naturally by taking a time average of the min-cut values between $S$ and $T$.
\end{proof}

The key difference between Theorem \ref{thm:static} and Theorem \ref{thm:average_mincut} is that in Theorem \ref{thm:static}, any failure pattern $f\in \m{F}$ may change the network topology as well as min-cut but $mincut(S,T) \geq |\m{X}(S)|$ holds for all $f\in \m{F}$ -- \ie $(G_f, \m{C})$ is assumed to be solvable. However, in Theorem \ref{thm:average_mincut}, we make no assumption about the connection as we are evaluating the average value of the min-cut.

Unlike the case of static ADT networks, with random erasures, it is necessary to maintain a queue at each node in the ADT network. This is because, if a link fails when a node has data to transmit on it, then it will have to wait until the link recovers. In addition, a transmitting node needs to be able to learn whether a packet has been received by the next hop node, and whether it was innovative -- this can be achieved using channel estimation, feedback and/or redundancy. In the original ADT network, the issue of feedback was removed by assuming that the links are noiseless bit-pipes. We present the following corollaries under these assumptions.

\begin{corollary}[Multicast in Time-varying Network]\label{thm:average_multicast} Consider an ADT network $G$ and a multicast connection $\m{C} = \{(S, T_1, \m{X}(S)), ..., (S, T_N, \m{X}(S))\}$. Assume that failures occur where failure patten $f \in \m{F}'$ occurs with probability $p_f$. Then, the multicast connection is feasible if and only if $mincut_{\m{F}'}(S,T_i) \geq |\m{X}(S)|$ for all $i$.
\end{corollary}
\begin{proof} Reference \cite{reliable} shows that the multicast connection is feasible if and only $mincut_{\m{F}'}(S,T_i) \geq |\m{X}(S)|$ for all $i$. The proof in \cite{reliable} relies on the fact that every node behaves like a stable $M/M/1$ queuing system in steady-state, and thus, the queues (or the number of innovative packets to be sent to the next hop node) has a finite mean if the network is run for sufficiently long period of time.
\end{proof}
\begin{corollary}[Random Coding for Time-varying Network] Consider $(G, \m{C})$ problem where $\m{C}$ is a multicast connection. Assume failure pattern $f \in \m{F}'$ occurs with probability $p_f$. Then, random linear network coding, where some or all code variables $\alpha_{(i, e_j)}, \beta_{(e_i, e_j)}, \epsilon_{(e_i, (D_j, k))}$ are chosen over all elements of $\mathbb{F}_q$ guarantees decodability at destination nodes $T_i$ for all $i$ simultaneously with arbitrary small error probability.
\end{corollary}
\begin{proof} This is a direct consequence of Corollary \ref{thm:average_multicast} and results in \cite{rlc}\cite{reliable}.
\end{proof}

%

\section{Network with Cycles}\label{sec:delay}

\begin{figure*}[tbp]
\[\scriptsize
\left(
\begin{array}{cccccccccccc}
1 & 0 & D & 0 & 0 & D & D^2\beta_{(e_3, e_7)} & 0 & D^2 \beta_{(e_6, e_9)} & D^2\beta_{(e_6, e_{10})} & D^3 \beta_{(e_6, e_9)} & D^3\beta_{(e_3, e_7)} + D^3 \beta_{(e_6, e_{10})}\\
0 & 1 & 0 & D & 0 & 0 & D^2 \beta_{(e_4, e_7)} & 0 & 0 & 0 & 0 & D^3 \beta_{(e_4, e_7)}\\
0 & 0 & 1 & 0 & 0 & 0 & D\beta_{(e_3, e_7)} &  0 & 0 & 0 & 0 & D^2 \beta_{(e_3, e_7)}\\
0 & 0 & 0 & 1 & 0 & 0 & D \beta_{(e_4, e_7)} & 0 & 0 & 0 & 0 & D^2 \beta_{(e_4, e_7)}\\
0 & 0 & 0 & 0 & 1 & 0 & 0 & 0 & 0 & 0 & 0 & 0\\
0 & 0 & 0 & 0 & 0 & 1 & 0 & 0 & D\beta_{(e_6, e_{9})} & D\beta_{(e_6, e_{10})}& D^2 \beta_{(e_6, e_9)} & D^2 \beta_{(e_6, e_{10})}\\
0 & 0 & 0 & 0 & 0 & 0 & 1 & 0 & 0 & 0 & 0 & D\\
0 & 0 & 0 & 0 & 0 & 0 & 0 & 1 & 0 & 0 & 0 & 0\\
0 & 0 & 0 & 0 & 0 & 0 & 0 & 0 & 1 & 0 & D & 0\\
0 & 0 & 0 & 0 & 0 & 0 & 0 & 0 & 0 & 1 & 0 & D\\
0 & 0 & 0 & 0 & 0 & 0 & 0 & 0 & 0 & 0 & 1 & 0\\
0 & 0 & 0 & 0 & 0 & 0 & 0 & 0 & 0 & 0 & 0 & 1\\
\end{array}
\right)
\]\vspace*{-.3cm}\caption{$12 \times 12$ matrix $(I-DF)^{-1}$ for network in Figure \ref{fig:network}. The matrix $F$ can be found in Figure \ref{fig:F}.}\vspace*{-.4cm}\label{fig:DF}
\end{figure*}

ADT networks are acyclic, with links directed from the source nodes to the destination nodes. However, wireless networks intrinsically have cycles as wireless links are bi-directional by nature. In this section, we extend the ADT network model to networks with cycles. In order to incorporate cycles, we need to introduce the notion of time -- since, without the notion of time, the network with cycles may not be casual. To do so, we introduce delay on the links. As in \cite{algebraic}, we model each link to have the same delay, and express the network random processes in the delay variable $D$.

We define $X_t(S, i)$ and $Z_t(T, j)$ to be the $i$-th and $j$-th binary random process generated at source $S$ and received at destination $T$ at time $t$, for $t = 1, 2,...$. We define $Y_t(e)$ to be the process on edge $e$ at time $t = 1, 2, ...$, respectively. We express the source processes as a power series in $D$, $\m{X}(S, D) = [X(S, 1, D), X(S, 2, D), ..., X(S, \mu(S), D)]$ where $X(S, i, D) = \sum_{t=0}^\infty X_t(S, i)D^t$. Similarly, we express the destination random processes $\m{Z}(T, D) = [Z(T, 1, D),$ $...,Z(T, \nu(Z), D)]$ where $Z(T, i, D) = \sum_{t=0}^\infty Z_t(T, i)D^t$. In addition, we express the edge random processes as $Y_t (e, D) = \sum_{t=0}^\infty Y_t(e)D^t$. Then, we can rewrite Equations (\ref{eq:y}) and (\ref{eq:y-s}) as
\begin{equation*}\label{eq:time_y-s}
Y_{t+1}(e) =\sum_{e'\in I(V)} \beta_{(e', e)} Y_t(e') + \sum_{X_t(S,i) \in \m{X}(S)} \alpha_{(i, e)} X_t(S, i).
\end{equation*}
Furthermore, the output processes $Z_t(T,i)$ can be rewritten as
\begin{align*}\label{eq:time_z}
Z_{t+1}(T, i)  =  \sum_{e' \in I(T)}\epsilon_{e',(T, i)} Y_{t}(e').
\end{align*}
%
Using this formulation, we can extend the results from \cite{algebraic} to ADT networks with cycles. We show that a system matrix $M(D)$ captures the input-output relationships of the ADT networks with delay and/or cycles.

\begin{theorem}\label{thm:delay}
Given a network $G = (\m{V}, \m{E})$, let $A(D)$, $B(D)$, and $F$ be the encoding, decoding, and adjacency matrices, as defined here:
\begin{align*}
A_{i,j} &=
\begin{cases}
\alpha_{(i, e_j)}(D) &\text{if $e_j \in O(S)$ and $X(S, i) \in \m{X}(S)$,}\\
0 &\text{otherwise.}
\end{cases}\\
B_{i, (T_j, k)} &=
\begin{cases}
\epsilon_{(e_i, (T_j, k))}(D) &\text{if $e_i \in I(T_j)$, $Z(T_j, k) \in \m{Z}(T_j)$,}\\
0 &\text{otherwise}.
\end{cases}
\end{align*}
and $F$ as in Equation (\ref{eq:F}). The variables $\alpha_{(i, e_j)} (D)$ and $\epsilon_{(e_i, (T_j, k))}(D)$ can either be constants or rational functions in $D$. Then, the system matrix of the ADT network with delay (and thus, may include cycles) is given as
\begin{equation}
M (D) = A(D) \cdot (I - DF)^{-1} \cdot B(D)^T.
\end{equation}
\end{theorem}
\begin{proof}The proof for this is similar to that of Theorem \ref{thm:m}; thus, we shall not discuss this in detail for want of space.
\end{proof}

Similar to Section \ref{sec:algebraic}, $(I-DF)^{-1}$ represents the impulse response of the network with delay. This is because the series $I + DF + D^2F^2 + D^3F^3 + ...$ represents the connectivity of the network while taking delay into account. For example, $F^k$ has a non-zero entry if there exists a path of length $k$ between two port. Now, since we want to represent the time associated with traversing from port $e_i$ to $e_j$, we use $D^kF^k$, where $D^k$ signifies that the path is of length $k$. Thus, $(I-DF)^{-1} = I + DF + D^2F^2 + D^3F^3 + ...$ is the impulse response of the network with delay. An example of $(I-DF)^{-1}$ for the example network in Figure \ref{fig:newnetwork} is shown in Figure \ref{fig:DF}.

Using the system matrix $M(D)$ from Theorem \ref{thm:delay}, we can extend Theorem \ref{thm:mincut_maxflow}, Theorem \ref{thm:singlemulticast}, Theorem \ref{thm:multiplemulticast}, Theorem \ref{thm:disjoint}, and Theorem \ref{thm:twolevel} to ADT networks with cycles/delay. However, there is a minor technical change. We now operate in a difference field -- instead of having coding coefficients from the finite field $\mathbb{F}_q$, the coding coefficients $\alpha_{(i, e_j)}(D)$ and $\epsilon_{((e_i, (T_j, k)))}(D)$ are now from $\mathbb{F}_q(D)$, the field of rational functions of $D$. We shall not discuss the proofs in detail; however, this is a direct application of results in \cite{algebraic}.

\section{Conclusions}
ADT networks \cite{adt1}\cite{adt2} have drawn considerable attention for its potential to approximate the capacity of wireless relay networks. In this paper, we showed that the ADT network can be described well within the algebraic network coding framework \cite{algebraic}. This connection between ADT network and algebraic network coding allows the use of results on network coding to understand better the ADT networks. We emphasize again that the aim of this paper is not to prove or disprove ADT network model's ability to approximate the capacity of the wireless networks, but to show that the ADT network problems, including that of computing the min-cut and constructing a code, can be captured by the algebraic network coding framework.

In this paper, we derived an algebraic definition of min-cut for the ADT networks, and provided an algebraic interpretation of the Min-cut Max-flow theorem for a single unicast/mulciast connection in ADT networks. Furthermore, by taking advantage of the algebraic structure, we have shown feasibility conditions for a variety of set of connections $\m{C}$, such as multiple multicast, disjoint multicast, and two-level multicast. We also showed optimality of linear operations for the connections listed above in the ADT networks, and showed that random linear network coding achieves capacity.

We extended the capacity characterization to networks with cycles and random erasures/failures. Thus, we proved the optimality of linear operations (as well as random linear network coding) for multicast connections in ADT networks with cycles. Furthermore, by incorporating random erasures into the ADT network model, we showed that random linear network coding is robust against failures and erasures.

\bibliography{References}
\bibliographystyle{IEEEtran}

\end{document}